\documentclass[amsmath,amssymb,notitlepage]{revtex4-1}

\usepackage{enumerate}
\usepackage{latexsym}
\usepackage{amsfonts}
\usepackage{amsthm}
\usepackage{color}
\usepackage{bbm}
\usepackage{dsfont}
\usepackage{graphicx}
\usepackage{textcomp}
\usepackage{enumerate}
\usepackage{graphicx,tikz}
\usetikzlibrary{arrows,matrix,calc}
\usepackage[colorlinks=true,linkcolor=blue,citecolor=blue,urlcolor=blue]{hyperref}
\usepackage{mathrsfs} % Calligraphic font
\newtheorem{theorem}{Theorem}

\newtheorem{lemma}{Lemma}

\newcommand{\cH}{\mathcal{H}}
\newcommand{\cT}{\mathcal{T}}

\DeclareMathOperator{\Tr}{Tr} %trace symbol
\DeclareMathOperator{\reals}{\mathbb{R}}
\DeclareMathOperator{\comp}{\mathbb{C}} 
\DeclareMathOperator{\iden}{\mathbbm{1}}

\newcommand{\norm}[1]{\left\lVert#1\right\rVert}
\usepackage{geometry}
\begin{document}
\title{Are almost-symmetries almost linear?}

\author{Javier Cuesta}
\email{j.cuesta@tum.de}
\affiliation{Department of Mathematics, Technische Universit\"at M\"unchen,  Garching, Germany}
\affiliation{Munich Center for Quantum
Science and Technology (MCQST),  M\"unchen, Germany}
\author{Michael M.~Wolf}
\email{m.wolf@tum.de}
\affiliation{Department of Mathematics, Technische Universit\"at M\"unchen,  Garching, Germany}
\affiliation{Munich Center for Quantum
Science and Technology (MCQST),  M\"unchen, Germany}
\date{\today}% It is always \today, today,

\vspace{-1cm}\begin{abstract}

It $d-$pends. Wigner's symmetry theorem implies that transformations that preserve transition probabilities of pure quantum states are linear maps on the level of density operators. We investigate the stability of this implication. On the one hand, we show that any transformation that preserves transition probabilities up to an additive $\varepsilon$ in a separable Hilbert space admits a weak  linear approximation, i.e. one relative to any fixed observable. This implies the existence of a linear approximation that is $4\sqrt{\varepsilon} d$--close in Hilbert-Schmidt norm, with $d$  the  Hilbert space dimension. On the other hand, we prove that  a  linear approximation that is close in  norm and independent of $d$ does not exist in general. To this end, we provide a lower bound that depends logarithmically on $d$.
\end{abstract}

\maketitle
\tableofcontents

\section{ Introduction and summary of results}

Wigner's theorem~\cite{Wig31} is a cornerstone for the mathematical representation of symmetries in quantum physics. It tells us that an arbitrary transformation on the set of pure states that preserves the ``transition probabilities'' must  necessarily correspond to a unitary or antiunitary operation. In particular, the transformation is representable by a \textit{linear map} on the level of density operators. Hence, Wigner's theorem is  arguably one of the reasons for the linear structure of quantum theory (besides various forms of locality\cite{Gisin,Jordan} and the probabilistic framework \cite{MM}).

 Wigner's theorem was proven in the general case, which does not assume bijectivity of the map, by Bargmann~\cite{Barg64} thirty years after Wigner's original idea. 
Recently~\cite{Geher14,SMCS14},  new and simpler proofs of this theorem have appeared where neither bijectivity of the map nor separability of the underlying complex Hilbert space $\mathcal{H}$ is assumed. If we denote by $\mathbb{P}(\mathcal{H})$ the set of pure states, identified with rank-one self-adjoint projections, Wigner's theorem reads:

\begin{theorem}[Wigner] Let $f:\mathbb{P}(\mathcal{H})\to \mathbb{P}(\mathcal{H})$ be a map that preserves transition probabilities, i.e. such that $\Tr f(X)f(Y)=\Tr XY$ for all $X,Y\in \mathbb{P}(\mathcal{H})$. Then there exists a linear or antilinear isometry $U:\mathcal{H}\to\mathcal{H}$ such that $f(X)=UXU^{\dagger}$.
\end{theorem}
The theorem can be seen to establish two things: linearity and isometry. In the present work, we focus on the stability of the linearity property. That is, we address the question: if a map between pure states almost preserves transition probabilities, how well can it be approximated by a linear map? 

Recently, a sequence of works~\cite{Ortho,Jordanauto,IT15} culminated in the result that Wigner's theorem is stable for finite-dimensional Hilbert spaces. We want to shed new light at least on the linear part of the problem and investigate in particular the role of the dimension of the underlying Hilbert space. To this end, we follow a different route than~\cite{Jordanauto,IT15} and employ (convex) geometry rather than analysis of operator algebras for the main argument. 
 Our results are two-fold. On the one hand, any map that approximately satisfies Wigner's symmetry condition in any separable Hilbert space is shown to admit a weak linear approximation. That is, when evaluated through an arbitrary but fixed observable, there exists a linear approximation even in case of infinite dimensional Hilbert spaces. As a corollary we obtain a linar approximation in Hilbert-Schmidt-norm whose accuracy scales linearly with dimension. This improves on the corresponding result of \cite{IT15}. 
 
In the second part, we address the problem from the other end and prove that a linear approximation in norm does not always exist in infinite dimensions---not even with respect to the operator norm. For that purpose, we study 
a componentwise logarithmic spiral map and prove that its operator norm distance to the set of linear maps essentially scales  logarithmically with the dimension of the Hilbert space. This holds despite the fact that the action of the map is arbitrarily close to that of a symmetry in Wigner's sense.

\section{Preliminaries}

We now introduce some notation and definitions. We denote by $\cH$ a complex Hilbert space, which we assume to be separable in the following. The space of bounded linear operators on $\mathcal{H}$ is denoted by $\mathfrak{B}(\mathcal{H})$ and its identity element by $\iden$. The adjoint of an operator $X$ is written as $X^{*}$. For $p\in[1,\infty)$ we denote by  $\mathcal{T}_{p}(\cH):=\{X\in\mathfrak{B}(\cH)|X=X^*, \norm{X}_{p}:=(\Tr |X|^{p} )^{1/p}  < \infty \}$  the real Banach space known as the hermitian $p-$Schatten class and its respective unit ball by $\mathcal{B}_{p}(\cH):=\{X\in\mathcal{T}_p(\cH)|\norm{X}_{p}\leq 1\}$. $\norm{\cdot}_\infty$ will be the operator norm on $\mathfrak{B}(\cH)$.

Occasionally, we will make use of the  Dirac ``bra-ket" notation where a vector in $\cH$ is written as $|x\rangle$ and the scalar product of two vectors $|x\rangle,|y\rangle$ as $\langle x|y\rangle$. A rank-one projection in $\mathfrak{B}(\cH)$ with range spanned by a unit vector $|x\rangle$ is then $|x\rangle\langle x|$. Using this we define $\mathbb{P}(\cH):=\{|\psi\rangle\langle\psi|\;|\;|\psi\rangle\in\cH, \norm{\psi}_{2}=1 \}$, which is the set of pure quantum states written as density operators.

\section{Almost-symmetries are close to linear} \label{sec:stability}
The main result of this section is summarized in the following theorem:

\begin{theorem}[Linear approximation of almost-symmetries]\label{th:stabwigner}
Let $\cH$ be a separable complex Hilbert space and $f:\mathbb{P}(\cH)\to\cT_2(\cH)$ a map satisfying  
\begin{equation}\label{eq:W}
\left| \Tr f(X)f(Y)-\Tr XY  \right|\leq \varepsilon \quad \text{for all }X,Y\in\mathbb{P}(\cH).
\end{equation}
\begin{enumerate}[(i)]
\item For any $A\in\mathcal{B}_{2}(\cH)$ there is a linear map $T_A:\cT_1(\cH)\to \cT_2(\cH) $ such that for all $X\in\mathbb{P}(\cH)$
\begin{equation*}
\left| \Tr\big[A\big(f(X)-T_A(X)\big)\big]\right|\leq  4\sqrt{\varepsilon}.
\end{equation*}
\item
If $\cH=\comp^d$, there exists a linear map $T:\cT_1(\cH)\to \cT_2(\cH) $ such that for all $X\in\mathbb{P}(\comp^{d})$
\begin{equation*}
\norm{f(X)-T(X)}_{2}\leq 4 d\sqrt{\varepsilon}.
\end{equation*}
\end{enumerate}
\end{theorem}

Note that Eq.~(\ref{eq:W}) is a slight relaxation of Wigner's condition since we allow $f$ to map into $\cT_2(\cH)$. That is, we do not restrict its range to the set of pure states. We will see that this generalization comes at no additional cost in the proof.

The overall strategy of the proof is the following: we first extend $f$ to a map $F$ that is defined on the entire space $\cT_1(\cH)$. Exploiting the condition in  Eq.~(\ref{eq:W}) we will show that $F$ almost preserves convex combinations. This enables the use of convex analysis and in particular of the geometric Hahn-Banach theorem to prove the existence of a linear approximation as stated in part (i) of the theorem. Part (ii) is then derived as a consequence of (i) by exploiting the existence of a finite basis.

We begin the proof of the theorem by extending the function $f$ to a larger domain. To this end,  we choose a spectral decomposition $X=\sum_{k}\lambda_{k}X_{k}$ for every $X\in\cT_ 1(\cH)$ for which $||X||_1=1$. Here $X_{k}\in\mathbb{P}(\cH)$ are assumed to be orthogonal with respect to the Hilbert-Schmidt inner product $\langle A, B\rangle:= \Tr A^{*}B$. Then we define $F:\cT_1(\cH)\to\cT_2(\cH)$ by extending
\begin{equation}\label{eq:F}
F(X):=\sum^{}_{k}\lambda_{k}f(X_{k}).
\end{equation} in a homogeneous way from the unit sphere to the entire space $\cT_1(\cH)$. Eq.~(\ref{eq:W}) then ensures that $\norm{F(X)}_2^2\leq\norm{X}_2^2+\varepsilon\norm{X}_1^2$ so that, indeed, $F(X)\in\cT_2(\cH)$.
By construction, $f$ is then the restriction of $F$ to the set $\mathbb{P}(\cH)$ of pure states and $F(\lambda X)=\lambda F(X)$ for all $\lambda\geq 0$. Moreover, Wigner's condition from Eq.~(\ref{eq:W}) easily extends to $F$:

\begin{lemma}[Wigner-condition for the extended map]\label{l:WF}
Let $f:\mathbb{P}(\cH)\to\cT_2(\cH)$ satisfy Eq.~\eqref{eq:W} and $F:\cT_1(\cH)\to\cT_2(\cH)$ be its  extension as defined above. Then
\begin{equation}\label{eq:WF}
\left| \langle F(X), F(Y)\rangle-\langle X,Y\rangle \right| \leq \norm{X}_{1}\norm{Y}_{1}\varepsilon, \qquad \text{for all }X,Y\in\cT_1(\cH).
\end{equation} 
\end{lemma}
\begin{proof}
Let $X=\sum_k\lambda_k X_k$ and $Y=\sum_j \mu_j Y_j$ be the spectral decompositions that define $F$ on $X,Y$ and recall that  for elements of $\cT_1(\cH)$ the trace-norm $||\cdot||_1$ is  the sum of the absolute values of eigenvalues. The Lemma then follows from applying  Eq.~\eqref{eq:W} to
\begin{equation*}
 \left| \big\langle \sum^{}_{k}\lambda_{k}f(X_{k}), \sum^{}_{j}\mu_{j}f(Y_{k})\big\rangle-\langle X,Y\rangle \right|\leq \sum^{}_{j,k}|\lambda_{k}||\mu_{j}| \left| \langle \;f(X_{k}), f(Y_{k})\;\rangle-\langle X_{k},Y_{j}\rangle \right|.
\end{equation*}
\end{proof}
From here we can show that $F$ is almost-linear in the following sense:
\begin{lemma}[Almost-linearity of the extended map]\label{l:quasilinear}
Let $F:\cT_1(\cH)\to\cT_2(\cH)$ be any map satisfying Eq.~\eqref{eq:WF}. Then for all $m\in\mathbb{N}, \lambda\in \reals^{m}$ and $X_1,\ldots,X_m\in \cT_1(\cH)$ we have 
\begin{equation*}
\norm{\sum^{m}_{i=1}\lambda_{i}F(X_{i})-F\left(\sum^{m}_{i=1}\lambda_{i}X_{i}\right)}_{2}\leq 2\sqrt{\varepsilon} \sum^{m}_{i=1}|\lambda_{i}|\norm{X_{i}}_{1}.
\end{equation*}
\end{lemma}

\begin{proof}
Consider $Z\in\cT_1(\cH)$ and use Eq.~\eqref{eq:WF} to bound 
\begin{align*}
&\left| \langle \sum^{m}_{i=1}\lambda_{i}F(X_{i})-F\left(\sum^{m}_{i=1}\lambda_{i}X_{i}\right), F(Z)\rangle \right|
 \\
&\leq \left| \sum^{m}_{i=1}\lambda_{i}\left(\langle F(X_{i}), F(Z)\rangle -\langle X_{i}, Z\rangle\right)  \right| + \left| \langle F\left(\sum^{m}_{i=1}\lambda_{i}X_{i}\right), F(Z)\rangle -\langle \sum^{m}_{i=1}\lambda_{i}X_{i}, Z\rangle  \right|, \\
&\leq\left( \sum^{m}_{i=1}\norm{\lambda_{i}X_{i}}_{1} +\norm{\sum^{m}_{i=1}\lambda_{i}X_{i}}_{1}\right)\norm{Z}_{1}\varepsilon\  \leq\ 2\varepsilon \norm{Z}_{1} \sum^{m}_{i=1}\norm{\lambda_{i}X_{i}}_{1}.
\end{align*}
Then by linearity of the Hilbert-Schmidt inner product
\begin{align*}
\norm{\sum^{m}_{i=1}\lambda_{i}F(X_{i})-F\left(\sum^{m}_{i=1}\lambda_{i}X_{i}\right)}^{2}_{2}&\leq \left| \langle \sum^{m}_{i=1}\lambda_{i}F(X_{i})-F\left(\sum^{m}_{i=1}\lambda_{i}X_{i}\right),\sum^{m}_{i=1}\lambda_{i}F(X_{i})\rangle\right|\\
&\quad+\left| \langle \sum^{m}_{i=1}\lambda_{i}F(X_{i})-F\left(\sum^{m}_{i=1}\lambda_{i}X_{i}\right),F\left(\sum^{m}_{i=1}\lambda_{i}X_{i}\right)\rangle\right|, \\
&\leq 2 \left(\sum^{m}_{i=1}\norm{\lambda_{i}X_{i}}_{1}\right)^{2}\varepsilon+2 \sum^{m}_{i=1}\norm{\lambda_{i}X_{i}}_{1}\norm{\sum^{m}_{i=1}\lambda_{i}X_{i}}_{1}\varepsilon,\\
&\leq \left(2\sum^{m}_{i=1}\norm{\lambda_{i}X_{i}}_{1}\right)^{2}\varepsilon.
\end{align*}

\end{proof}

Now we have all prerequisites for the proof of the main theorem.

\begin{proof}[{Proof of Thm.}~\ref{th:stabwigner}]  To show part (i), we define $\delta:=2\sqrt{\varepsilon}$ and consider the action of  $\hat{F}(X):=\Tr [AF(X)]$ on the unit-ball $\mathcal{B}_1(\cH)$. If $\norm{\lambda}_{1}\leq 1, \norm{A}_2\leq 1$ and $\norm{X_{j}}_1\leq 1$, then Cauchy-Schwarz and Lemma \ref{l:WF} imply
\begin{align*}
\left\lvert \hat{F}\left(\sum^{m}_{j}\lambda_{j}X_{j}\right)-\sum^{m}_{j}\lambda_{j}\hat{F}\left(X_{j}\right) \right\lvert &= \left\lvert \Tr A\left(\sum^{m}_{j=1}\lambda_{j}F(X_{j})-F\left(\sum^{m}_{j=1}\lambda_{j}X_{j}\right) \right) \right\lvert \\
&\leq \norm{A}_{2}\norm{\sum^{m}_{j=1}\lambda_{j}F(X_{j})-F\left(\sum^{m}_{j=1}\lambda_{j}X_{j}\right)}_{2} \ \leq\ \delta.
\end{align*}
Thus 
\begin{equation}\label{eq:lfW}
\sum^{m}_{j}\lambda_{j}\hat{F}\left(X_{j}\right) -\delta  \leq \hat{F}\left(\sum^{m}_{j}\lambda_{j}X_{j}\right)\leq  \sum^{m}_{j}\lambda_{j}\hat{F}\left(X_{j}\right) +\delta .
\end{equation}
Let $g_{-}$ and $g_{+}$ be the convex and concave envelopes of $\hat{F}$ over $\mathcal{B}_1(\cH)$. These are defined  as 
\begin{align*}
g_{-}(X)&:=\inf \left\{ \sum^{n}_{j=1}\lambda_{j}\hat{F}(X_{j})\; \big|\; X=\sum^{n}_{j=1}\lambda_{j}X_{j} \right\}, \\
g_{+}(X)&:=\sup \left\{ \sum^{n}_{j=1}\lambda_{j}\hat{F}(X_{j})\; \big|\; X=\sum^{n}_{j=1}\lambda_{j}X_{j} \right\},
\end{align*} 
taken over all finite convex decompositions of $X$ within the unit-ball $\mathcal{B}_1(\cH)$.
Using Eq.~\eqref{eq:lfW}, one verifies
\begin{equation}\label{eq:ube3}
g_{+}(X)-\delta \leq \hat{F}(X) \leq g_{-}(X)+\delta\quad \text{for all}\ X\in\mathcal{B}_{1}(\cH).
\end{equation}
Let $\Lambda_+$ and $\Lambda_-$ denote the subgraph of $X\mapsto g_{+}(X)-\delta$ and the supergraph of $X\mapsto g_{-}(X)+\delta$, respectively. Since $g_-$ and $-g_+$ are convex, $\Lambda_\pm$ are convex subsets of the direct-sum Banach space $\cT_1(\cH)\oplus\reals$. By construction they have non-empty interiors and due to Eq.~(\ref{eq:ube3}) the interiors are non-intersecting, i.e., ${\rm Int}(\Lambda_+)\cap {\rm Int}(\Lambda_-)=\emptyset$. By the geometric Hahn-Banach theorem, $\Lambda_+$ and $\Lambda_-$ can be separated by a closed hyperplane (cf. Fig.\ref{fig:1}b). Since, due to convexity, $\overline{{\rm Int}(\Lambda_\pm)}=\overline{\Lambda_\pm}$, this implies that there exists a continuous affine map $h:\cT_1(\cH)\to\reals$ such that for all $X\in\mathcal{B}_{1}(\cH)$, $g_{+}(X)-\delta \leq h(X) \leq g_{-}(X)+\delta$. Using that $\hat{F}\leq g_{+}$ and $g_{-}\leq \hat{F}$, the previous inequality implies:
\begin{equation*}
-\delta \leq g_{+}(X)-\hat{F}(X)-\delta \leq h(X)-\hat{F}(X) \leq g_{-}-\hat{F}(X) +\delta \leq \delta,
\end{equation*}
so $|\hat{F}(X)-h(X)|\leq \delta$. As $F(0)=0$ we can choose $h$ linear at the cost of $|\hat{F}(X)-h(X)|\leq 2\delta$. Defining $T_{A}:\cT_1(\cH)\to \cT_2(\cH)$ as $T_A(X):=h(X) A/\norm{A}_2^2$ then completes the proof of part $(i)$.\vspace*{5pt}

\textit{Proof of part (ii):} Let $\{A_{j}\}_{j=1}^{d^{2}}$ be a Hilbert-Schmidt orthonormal basis of self-adjoint operators on $\mathcal{H}=\comp^d$ and $h_{j}:\cT_1(\cH)\to \reals$  the corresponding linear maps from part $(i)$. Define a linear map $T:=\cT_1(\cH)\to\cT_2(\cH)$, $T(X):=\sum_{j=1}^{d^{2}}h_{j}(X)A_{j}$. Then for any $A=\sum_{j}a_{j}A_{j}$,

\begin{align*}
\Tr A\left(F(X)-T(X)\right)&=\sum_{i} a_{i}\Tr A_{i}F(X)-\sum_{i,j} a_{i}h_{j}(X)\Tr A_{i}A_{j} \\
&=\sum_{i} a_{i}\left(\Tr A_{i}F(X)-h_{i}(X)\right)\\
&\leq \norm{a}_{1}2\delta \leq 2\delta d\norm{A}_{2}.
\end{align*}
Therefore 
\begin{equation*}
\norm{F(X)-T(X)}_{2}=\sup_{\norm{A}_{2}\leq 1} \Tr A(F(x)-T(x)) \leq 2\delta d.
\end{equation*}

\end{proof}

\begin{figure}[ttt]
	\centering
		\includegraphics[clip, width=1\textwidth]{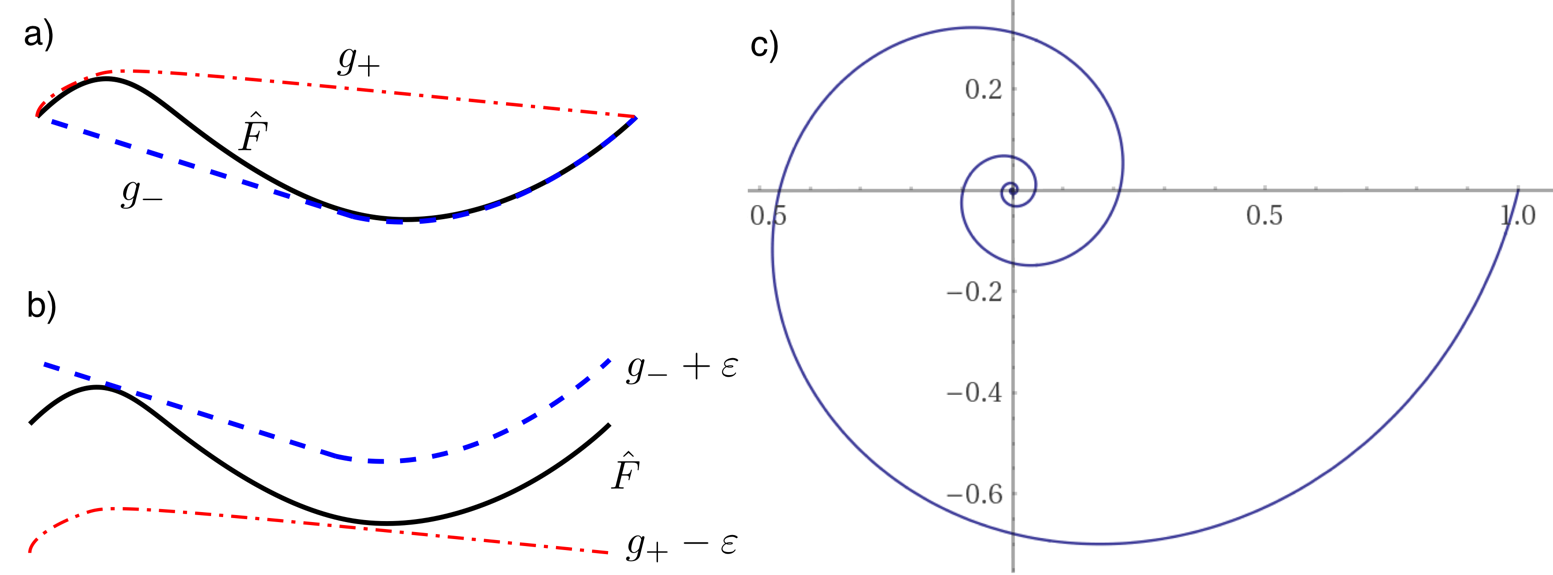}
	\caption{a) The convex ($g_-$) and concave ($g_+$) envelopes of $\hat{F}$. b) Shifting them by an appropriate $\varepsilon$, the corresponding supergraph and subgraph can be separated by a hyperplane, which then serves as a linear approximation of $\hat{F}$. c) The plot shows the image of the interval $[0,1]$ under the spiral map. Each point $z\in\comp$ undergoes a rotation around the origin by an angle that is proportional to $\ln |z|$. (For better visibility a large value  $\varepsilon=8$ is chosen for the plot, while $\varepsilon\sim 1/\ln d$ is considered in the proof.) }
	\label{fig:1}
\end{figure}

 \section{Almost-symmetries far from linear}\label{d_instability}
 
 In this section we will address the question from the other end and show that no linear approximation (as in Thm.\ref{th:stabwigner} ($ii$)) exists if the level of approximation is not allowed to depend on the dimension.  This is even true w.r.t. the operator norm, for which the intrinsic dimension dependence is minimal. The result is summarized in the following theorem where $S_{d}:=\{\psi\in\comp^{d}:\norm{\psi}_{2}=1 \}$ denotes the unit ball of $\comp^{d}$.

\begin{theorem}[Inapproximability]\label{thm:Dpendent}
Let $\varepsilon>0$ and $d\in\mathbb{N}$ such that $d\geq e^{\frac{4\pi}{\varepsilon}}+1$. There is a map $g:S_{d}\to S_{d}$ with the following properties:
\begin{enumerate}[(i)]
\item $\forall\psi,\phi\in S_{d}:$ \qquad $\left| |\langle g(\phi)|g(\psi)\rangle|^{2}-|\langle \phi|\psi\rangle|^{2}\right|\leq\varepsilon$.
\item For every linear map $T:\comp^{d\times d}\to \comp^{d\times d}$  we have: 
\begin{equation}
\sup_{\psi\in S_{d}}\norm{T(|\psi\rangle\langle\psi|)-|g(\psi)\rangle\langle g(\psi)|}_{\infty}\geq \frac{1}{3}.
\end{equation}
\end{enumerate}
\end{theorem}

\noindent Particular instances for $T$ would be $T(\cdot)=V\cdot V^{*}$ where $V$ is a unitary or anti-unitary on $\comp^{n}$. The heart of the proof of Thm.~\ref{thm:Dpendent} is the \textit{spiral map} (see Fig.\ref{fig:1}c) 
\begin{equation}\label{eq:spiralmap}
\comp\ni z\mapsto z|z|^{i\frac{\varepsilon}{2}}=ze^{i\frac{\varepsilon}{2}\ln|z|}.
\end{equation}
Its use  goes back at least to the work of Fritz John~\cite{John} and it has since then been used in various similar proofs, e.g. in Ref.~\onlinecite{Eva02,K91}.  It enters our discussion through the following:

\begin{lemma}\label{c:almostSym}
For any $\varepsilon >0$, let $g:\comp^d\rightarrow\comp^d$ be the map that acts as in Eq.(\ref{eq:spiralmap}) componentwise. 
\begin{enumerate}[(i)]
\item For all $\psi,\phi\in S_{d}$ we have
$
\left| |\langle g(\psi)|g(\phi)\rangle|^{2}-|\langle \psi|\phi\rangle|^{2}\right| \leq \varepsilon.
$
\item If $\varepsilon\ln (d-1)=4 \pi$, then $|\langle \varphi| g(\varphi) \rangle|=0$ holds for $|\varphi\rangle:=\frac{1}{\sqrt{2}}(1,\frac{1}{\sqrt{d-1}},\ldots,\frac{1}{\sqrt{d-1}})$.
\end{enumerate}

\end{lemma}
\begin{proof}
Let $\psi=(\psi_{1},\ldots,\psi_{d})$ and $\phi=(\phi_{1},\ldots,\phi_{d})$. Then using the inequalities $|e^{i\alpha}-1|\leq |\alpha|$ and $\ln x\leq x-1$ for $\alpha\in\reals$ and $x>0$ we obtain
\begin{align*}
\left| |\langle g(\psi)|g(\phi)\rangle|^{2}-|\langle \psi|\phi\rangle|^{2}\right| &\leq \sum_{k,l} |\phi_{k}\phi_{l}\psi_{k}\psi_{l}|\left|e^{i\frac{\varepsilon}{2}\ln\left|\frac{\phi_{k}\psi_{l}}{\phi_{l}\psi_{k}}\right|}-1 \right|,\\
&\leq \frac{\varepsilon}{2}\sum_{k,l} \left|\phi_{k}\phi_{l}\psi_{k}\psi_{l}\ln\left|\frac{\phi_{k}\psi_{l}}{\phi_{l}\psi_{k}}\right|\right|,\\
&\leq \frac{\varepsilon}{2}\sum_{k,l} \left||\phi_{k}\psi_{l}|^{2}-|\phi_{k}\phi_{l}\psi_{k}\psi_{l}|\right|,\\
&\leq \frac{\varepsilon}{2}\sum_{k,l} |\phi_{k}\psi_{l}|^{2}+|\phi_{k}\phi_{l}\psi_{k}\psi_{l}|,\\
&= \frac{\varepsilon}{2}\norm{\phi}^{2}_{2}\norm{\psi}^{2}_{2}+\frac{\varepsilon}{2}\left(\sum_{k}|\phi_{k}\psi_{k}|\right)^{2},\\
&\leq  \varepsilon\norm{\phi}^{2}_{2}\norm{\psi}^{2}_{2}= \varepsilon.
\end{align*}
Part $(ii)$ of the Lemma follows from inserting $\varepsilon\ln (d-1)=4 \pi$ into 
$$ |\langle\varphi|g(\varphi)\rangle|=\left|\frac12+\frac12\exp\Big[-\frac{i}{4}\varepsilon\ln(d-1)\Big]\right|.$$

\end{proof}

\begin{proof}[Proof of Theorem~\ref{thm:Dpendent}]
We use the spiral map $g:S_{d}\to S_{d}$ that acts componentwise as in Eq.~(\ref{eq:spiralmap}). If $d>\exp[4\pi/\varepsilon]+1$, then we decrease $\varepsilon$ until equality is achieved. In this way, Lemma \ref{c:almostSym} proves part $(i)$ of the theorem and at the same time guarantees that $g$ maps $\varphi$ onto an orthogonal vector.  In order to prove a bound on the best linear approximation, we exploit the symmetry of $g$. Let $G$ be the subgroup of $U(d)\subseteq\comp^{d\times d}$ that consists of all unitaries of the form $D\Pi$ where $D$ is a diagonal unitary and $\Pi$ a permutation matrix. Then $\forall\psi\in\comp^{d}: U^{-1}g(U\psi)=g(\psi)$ holds for all $U\in G$. The idea is now to argue that w.l.o.g. the best linear approximation has the same symmetry.

For every unitarily invariant norm on $\comp^{d\times d}$, in particular for the operator norm, and for any linear map $T:\comp^{d\times d}\rightarrow\comp^{d\times d}$ consider the following chain of inequalities:
\begin{align*}
\sup_{\psi\in S_d}\norm{T(|\psi\rangle\langle\psi|)-|g(\psi)\rangle\langle g(\psi)|}&=\sup_{\psi\in S_d}\norm{UT(U^{*}|\psi\rangle\langle\psi|U)U^{*}-|g(\psi)\rangle\langle g(\psi)|} \\
&\geq \sup_{\psi\in S_d} \int \norm{UT(U^{*}|\psi\rangle\langle\psi|U)U^{*}-|g(\psi)\rangle\langle g(\psi)|} dU \\
&\geq \sup_{\psi\in S_d}  \norm{\int UT(U^{*}|\psi\rangle\langle\psi|U)U^{*} dU-|g(\psi)\rangle\langle g(\psi)|},
\end{align*}
where $U\in G$, $dU$ is the Haar measure of $G$, and the first inequality uses $\sup \sum_{k}g_{k}\leq \sum_{k} \sup g_{k}$.  Following these inequalities, we can lower bound the quality of approximation of any linear map $T$ by the one of its symmetrized counterpart
\begin{equation*}
T_{G}(A):=\int_{G} UT(U^{*}AU)U^{*}dU.
\end{equation*}
As proven in Lemma 1 of Ref.~\cite{Anna} any linear map with this symmetry is specified by three parameters $\alpha,\beta,\gamma\in\comp$ and has the form 
\begin{equation}\label{eq:avT}
T_{G}(A)=\alpha \Tr [A] \iden +\beta A + \gamma \operatorname{diag}(A), 
\end{equation}
where $\operatorname{diag}(A)$ is the diagonal part of the matrix $A$. (Strictly speaking, Ref.~\cite{Anna} considers quantum channels, but since the relevant commutant is a vector space that is closed under taking adjoints, the parametrization in Eq.~(\ref{eq:avT}) holds for all linear maps.)

Using the state $|\varphi\rangle=\frac{1}{\sqrt{2}}(1,\frac{1}{\sqrt{d-1}},\ldots,\frac{1}{\sqrt{d-1}})$ from Lemma \ref{c:almostSym} for which $\langle\varphi|g(\varphi)\rangle=0$, we can bound
\begin{align*}
& \sup_{\psi\in S_d}\norm{T(|\psi\rangle\langle\psi|)-|g(\psi)\rangle\langle g(\psi)|}_{\infty}\geq \norm{T_{G}(|\varphi\rangle\langle\varphi|)-|g(\varphi)\rangle\langle g(\varphi)|}_{\infty} \\
&\geq \max\big\{ |\langle g(\varphi)|\left(T_G(|\varphi\rangle\langle\varphi|)-|g(\varphi)\rangle\langle g(\varphi)|\right) |g(\varphi)\rangle|, |\langle \varphi|T_G(|\varphi\rangle\langle\varphi|)-|g(\varphi)\rangle\langle g(\varphi)| |\varphi\rangle| \big\}, \\
&= \max\left\{\left|\alpha+\gamma\frac{d}{4(d-1)}-1 \right|, \left|\alpha+\beta+\gamma\frac{d}{4(d-1)}\right| \right\}\ \geq\ \frac{|\beta+1|}{2},
\end{align*}
where the last step used that for $x,y\in\comp$, $\max\{|x|,|y|\}\geq (|x|+|y|)/2 \geq |x-y|/2$. In order to eventually arrive at a parameter-independent lower bound we need a second inequality in which $\beta$ appears in a different way. For that purpose, let us denote the matrix of ones by $J_{d}\in\reals^{d\times d}$, $J_{ij}=1$, and the projection $P:=\iden-|1\rangle\langle 1|$. Since the operator-norm is sub-multiplicative we can obtain another lower bound via
\begin{align*}
\norm{T_{G}(|\varphi\rangle\langle\varphi|)-|g(\varphi)\rangle\langle g(\varphi)|}_{\infty}&\geq \norm{P(T_{G}(|\varphi\rangle\langle\varphi|)-|g(\varphi)\rangle\langle g(\varphi)|)P}_{\infty}, \\
&= \norm{\left(\alpha+\frac{\gamma}{2(d-1)}\right)\iden_{d-1}+\frac{\beta-1}{2(d-1)}J_{d-1}}_{\infty}, \\
&= \max\left\{\left|\alpha+\frac{\gamma}{2(d-1)}\right|,\left|\alpha+\frac{\gamma}{2(d-1)}+\frac{\beta-1}{2}\right|\right\}, \\
&\geq \frac{|\beta-1|}{4}.
\end{align*}
Finally, combining the two $\beta$-dependent bounds, we obtain
$$
 3 \sup_{\psi\in S_d}\norm{T(|\psi\rangle\langle\psi|)-|g(\psi)\rangle\langle g(\psi)|}_{\infty} \geq \frac{|\beta+1|}{2} + 2\; \frac{|\beta-1|}{4} 
 \geq \frac{\beta+1}{2} - \frac{\beta-1}{2}\ =\ 1.
$$
\end{proof}

%%%%%%%%%%%%%%%%%%%%%%%%%%%%%%%%%%%%%%%%%%%%%%%%%%%%%%%%%%%%%%%%%%%%%%%%%%%%%%%%%%%%%%%%%%%%%%%%%%%%%
\section{Discussion}\label{Discussion}

The inapproximability result of Thm.~\ref{thm:Dpendent} shows that a dimension-independent linear approximation result is not possible. This raises the question about the optimal dimension-dependence of a positive result of the form in Thm.~\ref{th:stabwigner} ($ii$). Thm.~\ref{thm:Dpendent}  imposes a logarithmic lower bound in the following way: 

For any map $f:\mathbb{P}(\comp^{d})\to\mathcal{T}_{2}(\comp^d)$ that fulfills Wigner's condition up to $\varepsilon$ according to Eq.~(\ref{eq:W}) define $\Delta(f):=\inf_{T}\sup_{\psi\in S^{d}}\norm{T(|\psi\rangle\langle\psi|)-f(|\psi)\rangle\langle \psi)|)}_{\infty}$ where the infimum is taken over all linear maps $T:\comp^{d\times d}\rightarrow\comp^{d\times d}$. Assume that $\sup_{f}\Delta(f)\leq \kappa(d)\varepsilon^{p}$ for some function $\kappa$ and some $p>0$. Choosing  $\varepsilon=4\pi/\ln(d-1)$ for sufficiently large $d$, Thm.~\ref{thm:Dpendent} provides a map $g$ that fulfills Eq.~\eqref{eq:W} together with $\Delta(g)\geq 1/3$. Therefore
\begin{align*}
\kappa(d) &\geq \frac{\sup_{f}\Delta(f)}{\varepsilon^{p}}, \\
 &\geq \frac{1}{3}\left(\frac{\ln(d-1)}{4\pi}\right)^{p}.
\end{align*}
On the other hand, Thm.~\ref{th:stabwigner} guarantees that $\sup_{f}\Delta(f)\leq 4d \sqrt{\varepsilon}$ so that  a significant gap between upper and lower bound remains. In order to close this gap, more sophisticated tools from  Banach space theory might be useful (see end remark in Ref.~[\onlinecite{K91}]). \vspace*{7pt}

\emph{Acknowledgments:} MMW acknowledges funding by the Deutsche Forschungsgemeinschaft (DFG, German Research Foundation) under Germany’s Excellence Strategy – EXC-2111 – 390814868.

%%%%%%%%%%%%%%%%%%%%%%%%%%%%%%%%%%%%%%%%%%%%%%%%%%%%%%%%%%%%%%%%%%%%%%%%%%%%%%%%%%%%%%%%%%%%%%%%%%%%%%%%%%%%%%%%%%%%%%%%%%%%%%%%%%%%%%%%%%%%%%
%%%%%%%%%%%%%%%%%%%%%%%%%%%%%%%%%%%%%%%%%%%%%%%%%%%%%%%%%%%%%%%%%%%%%%%%%%%%%%%%%%%%%%%%%%%%%%%%%%%%%%%%%%%%%%%%%%%%%%%%%%%%%%%%%%%%%%%%%%%%%%

% \bibliographystyle{natbib}
%\bibliography{GBSbibliography}

\end{document}